\documentclass[10pt]{article}

\usepackage{amsmath}
\usepackage{amsthm}
\usepackage{amssymb}

\theoremstyle{plain}
\newtheorem{thm}{Theorem}[section]
\newtheorem{lem}[thm]{Lemma}
\newtheorem{crl}[thm]{Corollary}

\theoremstyle{definition}
\newtheorem{defn}{Definition}[section]

\theoremstyle{remark}
\newtheorem{rem}{Remark}[section]

\numberwithin{equation}{section}

\newcommand{\Xa}{X_{\alpha}}
\newcommand{\Xb}{X_{\beta}}
\newcommand{\xa}{x_{\alpha}}
\newcommand{\xb}{x_{\beta}}
\newcommand{\mxa}{M_{X_{\alpha}}}

\newcommand{\Pa}{P_{\alpha}}
\newcommand{\Pdota}{\dot{P}_{\alpha}}
\newcommand{\Pddota}{\ddot{P}_{\alpha}}
\newcommand{\Pb}{P_{\beta}}
\newcommand{\pa}{p_{\alpha}}
\newcommand{\pb}{p_{\beta}}

\newcommand{\Da}{D_{\alpha}}
\newcommand{\fa}{f_{\alpha}}
\newcommand{\ut}{{\tilde u}}
\newcommand{\vt}{{\tilde v}}

\newcommand{\cthree}{C_0^3(\mathbb{R}^n)}
\newcommand{\ctwo}{C_0^2(\mathbb{R}^n)}
\newcommand{\eltwo}{L^2(\mathbb{R}^n)}

\newcommand{\hm}{\hspace{-2mm} &}
\newcommand{\dfr}[2]{\displaystyle \frac{#1}{#2}}
\newcommand{\ba}{\begin{array}}
\newcommand{\ea}{\end{array}}

\begin{document}

\title{\textbf{CHARACTERIZATION \\OF POINT TRANSFORMATIONS \\IN QUANTUM MECHANICS}}


\vspace{3cm}

\author{Yoshio Ohnuki and Shuji Watanabe}

\date{}

\maketitle

\begin{abstract}
We characterize point transformations in quantum mechanics from the mathematical viewpoint. To conclude that the canonical variables given by each point transformation in quantum mechanics correctly describe the extended point transformation, we show that they are all selfadjoint operators in $L^2(\mathbb{R}^n)$ and that the continuous spectrum of each coincides with $\mathbb{R}$. They are also shown to satisfy the canonical commutation relations.
\end{abstract}




\section{Introduction}

In classical mechanics the coordinate transformation
\begin{equation}\label{eq:mapping}
\left\{\ba{ll}
 f : \, x=(x_1,\,x_2, \,\ldots,\,x_n) \mapsto X=(X_1,\,X_2, \,\ldots,\,X_n),\\[1mm]
 \Xa = f_{\alpha}(x) \qquad (\alpha = 1, 2, \ldots, n)
\ea\right.
\end{equation}
is called a point transformation, where
\begin{equation}\label{eq:domain}
x \in D^n
\end{equation}
and the existence of $f^{-1}$ is assumed. In classical mechanics the domain $D^n$ does not always coincide with $\mathbb{R}^n$; it is sufficient for $D^n$ to involve the trajectory of a physical system under consideration.

It is known that the point transformation can be extended to a canonical transformation (see e.g. Whittaker \cite[p.293]{whittaker})
\[
(x_1,\ldots,x_n,p_1,\ldots,p_n) \mapsto (X_1,\ldots,X_n,P_1,\ldots,P_n),
\]
which is called an extended point transformation and is given by
\begin{equation}\label{eq:pt}\left\{\ba{ll}
 \Xa = \fa (x),\\[2mm]
 \Pa = {\displaystyle \sum_{\beta=1}^n}\, \dfr{\partial \xb}{\partial \Xa}\pb.
\ea\right.
\end{equation}
Here the canonical momenta $\pa$ and $\Pa$ are conjugate to $\xa$ and $\Xa$, respectively. Let $[A,\,B]_{\rm cl}$ stand for the classical Poisson bracket for $A(x,\,p)$ and $B(x,\,p)$:
\[
[A,\,B]_{\rm cl} = \sum_{\alpha = 1}^n\left(
\frac{\,\partial A\,}{\,\partial\xa\,}\frac{\,\partial B\,}{\,\partial\pa\,}
- \frac{\,\partial B\,}{\,\partial\xa\,}\frac{\,\partial A\,}{\,\partial\pa\,}
\right).
\]
The canonical variables $\xa$ and $\pa$ obey the relations
\[
[\xa ,\, \pb ]_{\rm cl}=\delta_{\alpha\beta}, \quad [\xa,\,\xb]_{\rm cl}=[\pa ,\,\pb ]_{\rm cl}=0.
\]
Then it is known that the new canonical variables $\Xa$ and $\Pa$ also obey
\begin{equation}\label{eq:newpb}
[X_\alpha,\,P_\beta]_{\rm cl}=\delta_{\alpha\beta}, \quad [X_\alpha,\,X_\beta]_{\rm cl}=[P_\alpha,\,P_\beta]_{\rm cl}=0.
\end{equation}

As an example, let us consider the point transformation $f: (x_1,\,x_2) \mapsto (r,\,\theta)$ from cartesian to plane polar coordinates. Here $(x_1,\,x_2) \in \mathbb{R}^2$. The existence of $f$ together with $f^{-1}$ implies $(x_1,\,x_2) \in D^2 = \mathbb{R}^2 \setminus \{ (0,\, 0) \}$. We are thus led to the extended point transformation $(x_1,\,x_2,\,p_1,\,p_2) \mapsto (r,\,\theta,\,p_r,\,p_{\theta})$. Here the canonical momenta $p_r$ and $p_{\theta}$ are conjugate to $r$ and $\theta$, respectively.

In quantum mechanics, however, the situation is quite different. It is known that the continuous spectrum of each canonical variable in quantum mechanics coincides with $\mathbb{R}$. Therefore, the point transformation $f: (x_1,\,x_2) \mapsto (r,\,\theta)$ from cartesian to plane polar coordinates is no longer allowed within the frame work of quantum mechanics. Hence the extended point transformation $(x_1,\,x_2,\,p_1,\,p_2) \mapsto (r,\,\theta,\,p_r,\,p_{\theta})$ is not allowed any longer. In fact, if it were allowed, then $r$, $\theta$, $p_r$ and $p_{\theta}$ would satisfy the canonical commutation relations. But this is not the case, because this clearly contradicts positivity of $r$ and boundedness of $\theta$.

The purpose of this paper is to characterize point transformations in quantum mechanics from the mathematical viewpoint. To this end we begin with defining a point transformation in classical mechanics and also that in quantum mechanics. Then, following DeWitt \cite{dewitt}, we define the new canonical momentum $\Pa$ conjugate to $\Xa$ in quantum mechanics. They \textit{should} be selfadjoint operators in a Hilbert space and the continuous spectrum of each \textit{should} coincide with $\mathbb{R}$. Moreover, they \textit{should} satisfy the canonical commutation relations. To conclude that the new canonical variables $\Xa$ and $\Pa$ correctly describe the extended point transformation in quantum mechanics, we show that they are all selfadjoint operators in $L^2(\mathbb{R}^n)$ and that the continuous spectrum of each of $\Xa$ and $\Pa$ coincides with $\mathbb{R}$. Moreover, we show that they satisfy the canonical commutation relations.

For simplicity, we use the unit $\hbar=1$ throughout this paper.


\section{Main results}

In classical mechanics the new canonical variables $\Xa$ and $\Pa$ given by \eqref{eq:pt} are required to obey \eqref{eq:newpb}, and hence the map $f$ is a $C^2$-diffeomorphism.

\begin{defn}\label{defn:ptcm}
We say that the map $f$ is a point transformation in classical mechanics if $f$ is a $C^2$-diffeomorphism and satisfies \eqref{eq:mapping}, \eqref{eq:domain}.
\end{defn}

\begin{rem}
Let $(x_1,\,x_2) \in \mathbb{R}^2$. The coordinate transformation $f: (x_1,\,x_2) \mapsto (r,\,\theta)$ from cartesian to plane polar coordinates is a point transformation in classical mechanics. The domain $D^2$ of the map $f$ does not contain the origin, i.e., $D^2 = \mathbb{R}^2 \setminus \{ (0, \, 0) \}$. Hence, $r$, $\theta$, $p_r$ and $p_{\theta}$ are canonical variables in classical mechanics:
\[
[r,\,p_r ]_{\rm cl}=[\theta,\,p_{\theta}]_{\rm cl}=1,\quad [r,\,\theta]_{\rm cl}=[r,\, p_{\theta}]_{\rm cl}=[\theta,\, p_r]_{\rm cl}=[p_r,\, p_{\theta}]_{\rm cl}=0.
\]
\end{rem}

In quantum mechanics the operators $\xa$ and $\pa$ are assumed to obey the canonical commutation relations
\[
[\xa ,\,\pb ]=i\,\delta_{\alpha\beta}, \quad [\xa ,\,\xb ]=[\pa ,\,\pb]=0,
\]
where $[A,\, B]=AB - BA$. Let $\xa$ be the multiplication by $\xa$. Then
\begin{equation}\label{eq:pa}
\pa = -i\,\frac{\partial}{\partial \xa}.
\end{equation}

\begin{defn}\label{defn:ptqm}
Let $f: \mathbb{R}^n \to \mathbb{R}^n$ be a bijective map satisfying
\begin{equation*}\left\{\ba{ll}
 f : \, x=(x_1,\,x_2, \,\ldots,\,x_n) \mapsto X=(X_1,\,X_2, \,\ldots,\,X_n),\\[1mm]
 \Xa = \fa (x) \qquad (\alpha = 1, 2, \ldots, n).
\ea\right.
\end{equation*}
We say that the map $f$ is a point transformation in quantum mechanics if $f$ is a $C^3$-diffeomorphism.
\end{defn}

\begin{rem}
Since the operator $\xa$ is the multiplication by $\xa$, the operator $\Xa$ is also the multiplication by $\fa (x)$.
\end{rem}

\begin{rem}
Let $(x_1,\,x_2) \in \mathbb{R}^2$. The coordinate transformation $f: (x_1,\,x_2) \mapsto (r,\,\theta)$ from cartesian to plane polar coordinates is \textit{not} a point transformation in quantum mechanics. This is because the domain of the map $f$ does not coincide with $\mathbb{R}^2$. Therefore, $r$, $\theta$, $p_r$ and $p_{\theta}$ are not canonical variables in quantum mechanics, and hence one can not impose the following relations
\[
[r,\,p_r ]=[\theta,\,p_{\theta}]=i,\quad [r,\,\theta]=[r,\, p_{\theta}]=[\theta,\, p_r]=[p_r,\, p_{\theta}]=0.
\]
\end{rem}

Let $\Pa$ $(\alpha = 1, 2, \ldots, n)$ denote the canonical momenta conjugate to $\Xa$. Following DeWitt, we define $\Pa$ in quantum mechanics.
\begin{defn}[DeWitt \cite{dewitt}]\label{defn:cvpt}
The new canonical variables $\Xa$ and $\Pa$ in quantum mechanics are defined by
\[
\left\{\ba{ll}
\Xa = \fa (x),\\[2mm]
\Pa = {\displaystyle \frac{1}{ \, 2 \, }\sum_{\beta=1}^n \left(
\frac{\,\partial \xb \,}{\,\partial \Xa\,}\,\pb +
\pb \, \frac{\,\partial \xb \,}{\,\partial \Xa\,}\right) }. 
\ea\right.
\]
\end{defn}

\begin{rem}
DeWitt \cite{dewitt} also showed that for $\Xa$ and $\Pa$, the operator $\pb$ is uniquely given by
\[
\pb = \frac{1}{ \, 2 \, }\sum_{\alpha = 1}^n
\left( \frac{ \, \partial \Xa \, }{ \, \partial \xb \, } \, \Pa +
\Pa \frac{ \, \partial \Xa \, }{ \, \partial \xb \, } \right).
\]
Interchanging, in this equality, the small and the capital letters we obtain the equality for $\Pa$ in Definition \ref{defn:cvpt}.
\end{rem}

\begin{rem}\label{rem:concrete}
Combining Definition \ref{defn:cvpt} with \eqref{eq:pa} yields
\[
\Pa = - i\,\sum_{\beta = 1}^n \frac{\partial \xb}{\partial \Xa}\frac{\partial}{\partial \xb} - \frac{i}{2}\sum_{\beta = 1}^n \frac{\partial}{\partial \xb}\left(\frac{\partial \xb}{\partial \Xa}\right).
\]
The operators $\Xa$ and $\Pa$ act on functions of $\xa$'s.
\end{rem}

We denote by $\sigma(A)$ (resp. by $\sigma_c(A)$) the spectrum (resp. the continuous spectrum) of an operator $A$. Let $\Pdota$ denote the operator $\Pa$ restricted $\cthree$ and let $\overline{\Pdota}$ denote the closure of $\Pdota$. Since the point transformation $f$ in quantum mechanics is a $C^3$-diffeomorphism (Definition \ref{defn:ptqm}), $\Pdota$ is a symmetric operator in $\eltwo$.
\begin{thm}\label{thm:sasp}
{\rm (a)}\quad The operator $\Xa$ is selfadjoint on $\displaystyle{D(\Xa) = \left\{ u(x) : u, \, \fa (x) u \in \eltwo \right\}}$.

{\rm (b)}\quad The set $\mathbb{R}$ coincides with the continuous spectrum of the selfadjoint operator $\Xa$ given by {\rm (a)}, i.e., $\sigma(\Xa) = {\sigma}_c(\Xa) = \mathbb{R}$.

{\rm (c)}\quad The operator $\Pdota$ is essentially selfadjoint.

{\rm (d)}\quad The set $\mathbb{R}$ coincides with the continuous spectrum of the selfadjoint operator $\Pa = \overline{\Pdota}$ given by {\rm (c)}, i.e., $\sigma(\Pa) = {\sigma}_c(\Pa) = \mathbb{R}$.
\end{thm}

We denote by $\Pa$ the selfadjoint operator $\overline{\Pdota}$ given by Theorem \ref{thm:sasp} (c).
\begin{thm}\label{thm:ccr}
The operators $\Xa$ and $\Pa$ satisfy the canonical commutation relations:
\[
[\Xa,\,\Pb]u = i\,{\delta}_{\alpha\beta}\, u, \quad [\Xa,\,\Xb]u = [\Pa,\,\Pb]u = 0
\]
{\it for} $u \in \ctwo$.
\end{thm}

We denote by $( \, \cdot \, , \, \cdot \, )$ the inner product of $\eltwo$:
\[
( u, \, v ) = \int_{\mathbb{R}^n} u(x) \, \overline{v(x)}\, dx \, ,\quad u, \, v \in \eltwo.
\]
We deal with only the case
\begin{equation}\label{eq:det}
J(x) = \frac{\partial(X_1,\ldots,X_n)}{\partial (x_1,\ldots,x_n) } > 0,
\end{equation}
since we can deal with another case $J(x) < 0$ in a similar manner.


\section{Proof of Theorem \ref{thm:sasp} (c)}

In this section we show that
\[
\ker ( {\Pdota}^{\ast} \pm i ) = \left\{ 0 \right\}
\]
to prove Theorem \ref{thm:sasp} (c). Here ${\Pdota}^{\ast}$ denotes the adjoint operator of $\Pdota$.

Let $u \in \ker ( {\Pdota}^{\ast} \pm i )$ and set
\begin{equation}\label{eq:a}
a(x) = \frac{1}{ \, 2 \, }\sum_{\beta = 1}^n \frac{\partial}{\partial \xb}
\left(\frac{ \, \partial \xb \, }{ \, \partial \Xa \, }\right) \pm 1.
\end{equation}
Then $a \in C^1(\mathbb{R}^n)$. Both $u$ and $a$ are functions of $\xa$'s, and hence of $\Xa$'s because of Definition \ref{defn:ptqm}. We now consider the following integral:
\begin{equation}\label{eq:integral}
\int_{\mathbb{R}^n} {J(x(X))}^{-1} \, u(x(X))
\left\{ \frac{\partial}{ \, \partial \Xa \, } + a(x(X)) \right\}
\overline{\Phi (X)} \, dX,
\end{equation}
where $\Phi \in \cthree$ is a function of $\Xa$'s.
\begin{lem}\label{lem:zero}
Let $u \in \ker ( {\Pdota}^{\ast} \pm i )$ and let $a$ be given by \eqref{eq:a}. Then for $\Phi \in \cthree$,
\[
\int_{\mathbb{R}^n} {J(x(X))}^{-1} \, u(x(X))
\left\{ \frac{\partial}{ \, \partial \Xa \, } + a(x(X)) \right\}
\overline{\Phi (X)} \, dX = 0.
\]
\end{lem}
\begin{proof}
The substitution $X = f(x)$ (see Definition \ref{defn:ptqm}) turns \eqref{eq:integral} into
\begin{eqnarray*}
\, & &\int_{\mathbb{R}^n} {J(x(X))}^{-1} \, u(x(X))
\left\{ \frac{\partial}{ \, \partial \Xa \, } + a(x(X)) \right\}
\overline{\Phi (X)}\, dX \\
&=& \int_{\mathbb{R}^n} u(x) \left\{ \sum_{\beta = 1}^n
\frac{ \, \partial \xb \, }{ \, \partial \Xa \, }
\frac{\partial}{\partial \xb} + a(x) \right\}\overline{\Phi (X(x))}\, dx.
\end{eqnarray*}
Set $\phi (x) = \Phi (X(x))$. Then $\phi \in \cthree = D( \Pdota )$. Therefore,
\[
\int_{\mathbb{R}^n} u(x) \left\{ \sum_{\beta = 1}^n
\frac{ \, \partial \xb \, }{ \, \partial \Xa \, }
\frac{\partial}{\partial \xb} + a(x) \right\}\overline{\Phi (X(x))}\, dx
= -i \left( u, \left( {\Pdota} \mp i \right) \phi \right) = 0.
\]
The lemma follows.
\end{proof}

Set
\begin{equation}\label{eq:AU}
U(X) = {J(x(X))}^{-1} \, u(x(X)) \quad {\rm and} \quad A(X) = a(x(X)),
\end{equation}
where $u \in \ker ( {\Pdota}^{\ast} \pm i )$ and $a$ is given by \eqref{eq:a}. Note that $U$ and $A$ are functions of $\Xa$'s. Let
\begin{equation}\label{eq:omega}
\Omega = \left( a_1, \, b_1 \right) \times \cdots \times
\left( a_{\alpha}, \, b_{\alpha} \right) \times \cdots \times
\left( a_n, \, b_n \right) \in \mathbb{R}^n,
\end{equation}
where $a_{\alpha}$, $b_{\alpha} \in \mathbb{R}$ and $| a_{\alpha} |, \, | b_{\alpha} | < \infty$ \  $(\alpha = 1, 2, \dots, n)$. Since $\Phi \in \cthree$, we let $\Phi \in C_0^{\infty}(\Omega)$. Here $\Omega$ is given by \eqref{eq:omega} and $X=(X_1,\,\ldots,\,\Xa,\,\ldots,\,X_n) \in \Omega$. Lemma \ref{lem:zero} also holds for $\Phi \in C_0^{\infty}(\Omega)$. Moreover, it is easy to see that
\[
U, \, AU \in L_{{\rm loc}}^1(\Omega).
\]
Lemma \ref{lem:zero} thus implies the following.
\begin{lem}\label{lem:weakd}
Let $U$ and $A$ be as in \eqref{eq:AU}. Then there is the weak derivative $\Da U$ of $U$ satisfying
\[
\Da U = AU, \qquad \Da = \frac{\partial}{ \, \partial \Xa \, } \, .
\]
\end{lem}

Let us regard $U$ and $A$ as functions of $\Xa$ only, where $a_{\alpha} < \Xa < b_{\alpha}$.
\begin{lem}\label{lem:regularity}
Let $U$ be given by \eqref{eq:AU} and let $\Omega$ be as in \eqref{eq:omega}. Then
\[
U \in C^1\left( a_{\alpha}, \, b_{\alpha}\right).
\]
\end{lem}
\begin{proof} \quad It is easy to see from \eqref{eq:det} that $U \in L^2(\Omega)$ since
\begin{eqnarray*}
\int_{\Omega} { \left| U(X) \right| }^2 \, dX &\leq&
\sup_{X \in \Omega}\left\{ {J(x(X))}^{-1} \right\}
\cdot \int_{\Omega} {J(x(X))}^{-1}{\left| \, u(x(X)) \, \right|}^2 \, dX \\
&=& \sup_{X \in \Omega}\left\{ {J(x(X))}^{-1} \right\}
\cdot \int_{f^{-1}(\Omega)} {\left| \, u(x) \, \right|}^2 \, dx \\
&<& \infty ,
\end{eqnarray*}
where $f^{-1}(\Omega)$ is the image of $\Omega$ under the inverse of the map $f$. Hence $U \in L^2 (a_{\alpha}, \, b_{\alpha})$, which implies $AU \in L^2 (a_{\alpha}, \, b_{\alpha})$. Therefore, $\Da U = AU \in L^2 (a_{\alpha}, \, b_{\alpha})$, and hence
\[
U \in H^1 (a_{\alpha}, \, b_{\alpha}).
\]
Since $A \in C^1(\mathbb{R}^n)$, it follows that $AU \in H^1 (a_{\alpha}, \, b_{\alpha})$. Therefore, $\Da U \in H^1 (a_{\alpha}, \, b_{\alpha})$. Thus
\[
U \in H^2 (a_{\alpha}, \, b_{\alpha}).
\]
The lemma follows from the Sobolev embedding theorem (see e.g. Goldstein \cite[p.135]{goldstein} or Reed and Simon \cite[p.52]{reedsimon}).
\end{proof}

Noting Lemmas \ref{lem:weakd} and \ref{lem:regularity} we now solve
\[
\Da U = AU,  \qquad \Da = \frac{\partial}{ \, \partial \Xa \, }
\]
to obtain the solution explicitly. It follows from \eqref{eq:a} and \eqref{eq:AU} that
\[
\left\{ \frac{\partial}{ \, \partial \Xa \, }
- \frac{1}{ \, 2 \, }\sum_{\beta = 1}^n \frac{\partial}{\partial \xb}
\left(\frac{ \, \partial \xb \, }{ \, \partial \Xa \, }\right) \mp 1 \right\}
{J(x)}^{-1} u(x(X)) = 0.
\]
Hence
\begin{equation}\label{eq:pde}
\left\{ \frac{\partial}{ \, \partial \Xa \, }
+ \frac{1}{ \, 2 \, }\sum_{\beta = 1}^n \frac{\partial}{\partial \xb}
\left(\frac{ \, \partial \xb \, }{ \, \partial \Xa \, }\right) \mp 1 \right\}
u(x(X)) = 0.
\end{equation}
Here we used the following lemma.
\begin{lem}\label{lem:cal}
\[
\left\{ \frac{\partial}{\partial \Xa} J(x) \right\} {J(x)}^{-1}
= - \sum_{\beta = 1}^n \frac{\partial}{\partial \xb}
\left(\frac{ \, \partial \xb \, }{ \, \partial \Xa \, }\right).
\]
\end{lem}
\begin{proof}
A straightforward calculation gives
\[
\frac{\partial}{\partial \Xa} J(x)
= \sum_{\lambda=1}^n \left| \begin{array}{cccccc}
\dfr{\partial X_1}{\partial x_1}\hm \dfr{\partial X_1}{\partial x_2}
\hm\cdots &\cdots\hm\dfr{\partial X_1}{\partial x_n}\\[2mm]
\vdots\hm\vdots\hm\vdots &\vdots\hm\vdots\\[2mm]
\dfr{\partial X_{\lambda-1}}{\partial x_1}\hm 
\dfr{\partial X_{\lambda-1}}{\partial x_2}
\hm\cdots &\cdots\hm\dfr{\partial X_{\lambda-1}}{\partial x_n}\\[5mm]
\dfr{\partial}{\partial \Xa}
\Bigl(\dfr{\partial X_{\lambda}}{\partial x_1}\Bigr)\hm
\dfr{\partial}{\partial \Xa}
\Bigl(\dfr{\partial X_{\lambda}}{\partial x_2}\Bigr)
\hm\cdots&\cdots\hm\dfr{\partial}{\partial \Xa}
\Bigl(\dfr{\partial X_{\lambda}}{\partial x_n}\Bigr)\\[5mm]
\dfr{\partial X_{\lambda+1}}{\partial x_1}\hm 
\dfr{\partial X_{\lambda+1}}{\partial x_2}
\hm\cdots\hm\cdots\hm\dfr{\partial X_{\lambda+1}}{\partial x_n}\\[2mm]
\vdots\hm\vdots\hm\vdots &\vdots\hm\vdots\\[2mm]
\dfr{\partial X_n}{\partial x_1}\hm\dfr{\partial X_n}{\partial x_2}
\hm\cdots\hm\cdots\hm\dfr{\partial X_n}{\partial x_n}
\end{array}\right| \, .
\]
Therefore,
\begin{eqnarray*}
\left\{ \frac{\partial}{\partial \Xa} J(x) \right\} {J(x)}^{-1}
&=& \sum_{\lambda, \beta=1}^n \frac{\partial}{\partial \Xa}
\left(\frac{\partial X_{\lambda}}{\partial \xb}\right)
\cdot\frac{\partial \xb}{\partial X_{\lambda}}\\[4mm]
&=& - \sum_{\lambda, \beta=1}^n \frac{\partial X_{\lambda}}{\partial \xb}
\cdot\frac{\partial}{\partial \Xa}
\left( \frac{\partial \xb}{\partial X_{\lambda}} \right) \\[4mm]
&=& - \sum_{\beta = 1}^n \frac{\partial}{\partial \xb}
\left(\frac{ \, \partial \xb \, }{ \, \partial \Xa \, }\right).
\end{eqnarray*}
\end{proof}

Let us continue the proof of Theorem \ref{thm:sasp} (c). Set
\[
u(x(X)) = \sqrt{ \, J(x(X))  \,} v(x(X)).
\]
Applying again Lemma \ref{lem:cal} turns \eqref{eq:pde} into
\[
\left( \frac{\partial}{ \, \partial \Xa \, } \mp 1 \right) v(x(X)) = 0.
\]
Therefore, $v(x(X)) = C(X_1, \dots, X_{\alpha - 1}, X_{\alpha + 1}, \dots, X_n )e^{\pm \Xa}$, and hence
\[
u(x(X)) = C(X_1, \dots, X_{\alpha - 1}, X_{\alpha + 1}, \dots, X_n )
\sqrt{ \, J(x(X)) \, } e^{\pm \Xa} \, .
\]
Here the function $C$ does not depend on $\Xa$. Since $u \in \eltwo$,
\begin{eqnarray*}
& & \int_{\mathbb{R}} e^{\pm 2 \Xa} \, d\Xa \times \\
& & \quad \times
\int_{\mathbb{R}^{n - 1}}
{\left| C(X_1, \dots, X_{\alpha - 1}, X_{\alpha + 1}, \dots, X_n ) \right|}^2
dX_1 \cdots dX_{\alpha - 1}dX_{\alpha + 1} \cdots dX_n < \infty \, .
\end{eqnarray*}
But
\[
\int_{\mathbb{R}} e^{\pm 2 \Xa} \, d\Xa = \infty \, .
\]
Hence $C(X_1, \dots, X_{\alpha - 1}, X_{\alpha + 1}, \dots, X_n ) = 0$,
which implies $u = 0$. Thus
\[
\ker ( {\Pdota}^{\ast} \pm i ) = \left\{ 0 \right\} .
\]

The proof of Theorem \ref{thm:sasp} (c) is complete.


\section{Proofs of the rest}

In this section we prove the rest of our main results from the viewpoint of unitary equivalence.

Let $u \in \eltwo$. Since each point transformation in quantum mechanics is a $C^3$-diffeomorphism, the function $x \mapsto u(x)/\sqrt{J(x)}$ can be regarded as a function of $\Xa$'s. We therefore set
\begin{equation}\label{eq:tilde}
\ut (X) = \frac{u(x(X))}{\,\sqrt{J(x(X))}\, },
\end{equation}
where $\ut$ is a function of $\Xa$'s. A straightforward calculation gives
\[
\int_{\mathbb{R}^n} \left| \ut (X) \right|^2 \, dX = \int_{\mathbb{R}^n} \left| u(x) \right|^2 \, dx < \infty.
\]
Hence $\ut \in \eltwo$. Let us define $U: \eltwo \to \eltwo$ by
\[
U: \, u \longmapsto \ut,
\]
where $u$ (resp. $\ut$) is a function of $\xa$'s (resp. of $\Xa$'s).

\begin{lem}\label{lem:property}
The operator $U: \eltwo \to \eltwo$ is unitary and satisfies the following.

{\rm (a)}\quad $U\cthree \subset \ctwo$,

{\rm (b)}\quad $U\ctwo = \ctwo$.
\end{lem}
\begin{proof}
The lemma follows immediately from Definition \ref{defn:ptqm} and \eqref{eq:tilde}.
\end{proof}

We denote by $\displaystyle{-i\frac{\partial}{\partial \Xa}}$ the operator
\[
-i\frac{\partial}{\partial \Xa}: \ut \longmapsto -i\frac{\partial \ut}{\partial \Xa}
\]
with domain
\begin{eqnarray*}
& &D(-i\frac{\partial}{\partial \Xa}) \\
&=& \left\{ \ut (X): \ut \;\mbox{is absolutely continuous with respect to} \;\Xa , \; \ut, \, \frac{\partial \ut}{\partial \Xa} \in \eltwo \right\}.
\end{eqnarray*}
Note that the operator $\displaystyle{-i\frac{\partial}{\partial \Xa}}$ acts on functions of $\Xa$'s. On the other hand, the operator $\Pa$ acts on functions of $\xa$'s (see Remark \ref{rem:concrete}).

Let $\Pddota$ denote the operator $\Pa$ restricted $\ctwo$ and let $\overline{\Pddota}$ denote the closure of $\Pddota$. A straightforward calculation gives the following.
\begin{lem}\label{lem:onctwo}
$\displaystyle{\Pddota u = U^{\ast}\left( -i\frac{\partial}{\partial \Xa} \right)U u}$, \quad $u \in \ctwo$.
\end{lem}

\begin{lem}\label{lem:pandd}
$\displaystyle{\overline{\Pddota} = U^{\ast}\left( -i\frac{\partial}{\partial \Xa} \right)U}$.
\end{lem}
\begin{proof} \quad {\it Step 1} \quad  We show
\[
\overline{\Pddota} \subset U^{\ast}\left( -i\frac{\partial}{\partial \Xa} \right)U.
\]
For $u \in D(\overline{\Pddota})$, there is a sequence $\{ u_m \}_m \subset \ctwo$ satisfying
\[
u_m \to u, \qquad \Pddota u_m \to \overline{\Pddota} u \quad \mbox{in}\; \eltwo.\]
Combining Lemma \ref{lem:property} with Lemma \ref{lem:onctwo} yields
\[
U u_m \to U u, \qquad \left( -i\frac{\partial}{\partial \Xa} \right)U u_m \to U \overline{\Pddota} u \quad \mbox{in}\; \eltwo.
\]
Lemma \ref{lem:property} (b) implies that $U u_m \in \ctwo \subset D(-i\,\partial / \partial \Xa)$. Since $-i\,\partial / \partial \Xa$ is a closed operator,
\[
U u \in D(-i\frac{\partial}{\partial \Xa}), \qquad U^{\ast}\left( -i\frac{\partial}{\partial \Xa} \right)U u = \overline{\Pddota} u.
\]

{\it Step 2} \quad We next show
\[
\overline{\Pddota} \supset U^{\ast}\left( -i\frac{\partial}{\partial \Xa} \right)U.
\]
For $v \in U^{\ast}D(-i\,\partial / \partial \Xa)$, we set $\vt = U v$. Then $\vt \in D(-i\,\partial / \partial \Xa)$. Let us recall here that
\[
\mbox{the operator} \; -i\frac{\partial}{\partial \Xa} \; \mbox{with domain} \; \ctwo \; \mbox{is essentially selfadjoint.}
\]
Hence there is a sequence $\{ {\vt}_m \}_m \subset \ctwo$ satisfying
\[
{\vt}_m \to \vt, \qquad -i\frac{\partial}{\partial \Xa} {\vt}_m \to -i\frac{\partial}{\partial \Xa} \vt \quad \mbox{in}\; \eltwo.
\]
Set $v_m = U^{\ast}{\vt}_m$. Then $v_m \in \ctwo$ by Lemma \ref{lem:property} (b). Combining Lemma \ref{lem:property} with Lemma \ref{lem:onctwo} again yields
\[
U^{\ast}{\vt}_m \to v, \qquad \Pddota U^{\ast} {\vt}_m \to U^{\ast}\left( -i\frac{\partial}{\partial \Xa} \right)U v \quad \mbox{in}\; \eltwo.
\]
Therefore,
\[
v \in D(\overline{\Pddota}), \qquad \overline{\Pddota} v = U^{\ast}\left( -i\frac{\partial}{\partial \Xa} \right)U v.
\]
The lemma follows.
\end{proof}

Theorem \ref{thm:sasp} (c) immediately implies the following.
\begin{crl}\label{crl:psa}
The operator $\Pddota$ is essentially selfadjoint.
\end{crl}

Combining Lemma \ref{lem:property} with Lemma \ref{lem:pandd} also implies Corollary \ref{crl:psa}. Moreover, it is easy to see that the following holds.
\begin{lem}\label{lem:relationp}
$\displaystyle{\overline{\Pdota} = \overline{\Pddota} = U^{\ast}\left( -i\frac{\partial}{\partial \Xa} \right)U}$.
\end{lem}

We denote by $\mxa$ the multiplication by$\Xa$:
\[
\mxa : \ut \longmapsto \Xa \ut
\]
with domain
\[
D(\mxa) = \left\{ \ut (X): \ut, \, \Xa \ut \in \eltwo \right\}.
\]
Note that the operator $\mxa$ acts on functions of $\Xa$'s. On the other hand, the operator $\Xa$ acts on functions of $\xa$'s (see Remark \ref{rem:concrete}). A straightforward calculation gives the following.
\begin{lem}\label{lem:relationx}
The operators $\Xa$ and $\mxa$ are unitarily equivalent, i.e.,
\[
\Xa = U^{\ast} \mxa U.
\]
\end{lem}

Each of Theorem \ref{thm:sasp} (a), (b), (d) and Theorem \ref{thm:ccr} immediately follows from Lemmas \ref{lem:relationp} and \ref{lem:relationx}.


\section{An example of a point transformation in quantum mechanics}

Let $x \in \mathbb{R}$. The coordinate transformation
\[
x \longmapsto X = \sinh x
\]
is a point transformation in quantum mechanics. The corresponding extended point transformation is
\[
\left( x, \, -i\,\frac{\partial}{\partial x} \right) \longmapsto \left( X = \sinh x, \, P = -i\,\frac{1}{\,\cosh x \,}\left( \frac{\partial}{\, \partial x \,} - \frac{1}{\, 2\,}\tanh x \right) \right).
\]
Therefore, the operators $X$ and $P$ acting on functions of $x$ are both canonical variables in quantum mechanics.


\vspace{0.5cm}

Y. OHNUKI \\
\noindent (former address) Nagoya Women's University \\
\noindent 1302 Takamiya, Tempaku, Nagoya 468-8507, Japan \\
\noindent e-mail address: ohnuki@nagoya-wu.ac.jp \\
\vspace{0.2cm} \\
S. WATANABE \\
\noindent (former address) Department of Electronics and Information Engineering \\
\noindent Aichi University of Technology \\
\noindent 50-2 Manori, Nishihazama-cho, Gamagouri 443-0047, Japan \\
\noindent e-mail address: watanabe@aut.ac.jp


\thebibliography{4}
\bibitem{dewitt}DeWitt, B.S., \textit{Point transformations in quantum mechanics}, Phys. Rev. \textbf{85} (1952), 653--661.

\bibitem{goldstein}Goldstein, J.A., \textit{Semigroups of Linear Operators and Applications}, Oxford University Press, New York, 1985/Clarendon Press, Oxford, 1985.

\bibitem{reedsimon}Reed, M. and Simon, B., \textit{Methods of Modern Mathematical Physics II}, \textit{Fourier Analysis, Self-Adjointness}, Academic Press, New York, 1975.

\bibitem{whittaker}Whittaker, E.T., \textit{A Treatise on the Analytical Dynamics of Particles and Rigid Bodies}, 4th ed., Cambridge University Press, Cambridge, New York, Melbourne, 1937.
\endthebibliography
\end{document}